\documentclass[11pt]{article}
\usepackage{graphics,latexsym,amsfonts}

\usepackage{amsmath}
\usepackage{amssymb}
\usepackage{latexsym}
\usepackage{epsfig}
\usepackage[mathscr]{eucal}
\usepackage{eufrak}
\usepackage{graphics} \usepackage{color} \usepackage{multicol}
\usepackage{enumerate} \usepackage{array}

\begin{document}
\newtheorem{lemma}{Lemma}
\newtheorem{definition}[lemma]{Definition}
\newtheorem{proposition}[lemma]{Proposition}
\newtheorem{theorem}[lemma]{Theorem}
\newtheorem{corollary}[lemma]{Corollary}
\newcommand{\bx}{\hfill \rule{2mm}{2.5mm}}

\newenvironment{proof}{\emph{Proof:}}{\bx}
\def\nl{\medskip \\ \noindent}
\def\np{

\medskip}
\setcounter{page}{1}
\normalsize

\title{\vspace{-5mm}{\bf Integer programming methods for special college admissions problems}}

\author{P\'eter Bir\'o$^1$\thanks{Visiting faculty at the Economics Department, Stanford University in year 2014. Supported by the Hungarian Academy of Sciences under its Momentum Programme (LD-004/2010) and also by OTKA grant no.\ K108673.} \ and
Iain McBride$^2$\thanks{Supported by a SICSA Prize PhD Studentship.
}
\\
\vspace{-2mm}
\\
\small
$^1$ \emph{Institute of Economics, Research Centre for Economic and Regional Studies,}
\vspace{-1mm}
\\
\small
\emph{Hungarian Academy of Sciences, H-1112, Buda\"orsi \'ut 45, Budapest, Hungary, and}
\vspace{-1mm}
\\
\small
\emph{Department of Economics, Stanford University}
\vspace{-1mm}
\\
\small
\emph{Email:} {\tt peter.biro@krtk.mta.hu}.
\\
\vspace{-2mm}
\\
\small
$^2$ \emph{School of Computing Science, University of Glasgow}
\vspace{-1mm}
\\
\small
\emph{Sir Alwyn Williams Building, Glasgow G12 8QQ, UK.}
\vspace{-1mm}
\\
\small
\emph{Email:} {\tt i.mcbride.1@research.gla.ac.uk}.
}
\date{ }
\maketitle

\vspace{-7mm}
\begin{quotation}
\small \noindent {\bf Abstract.}
We develop Integer Programming (IP) solutions for some special college admission problems arising from the Hungarian higher education admission scheme. We focus on four special features, namely the solution concept of stable score-limits, the presence of lower and common quotas, and paired applications. We note that each of the latter three special feature makes the college admissions problem NP-hard to solve. Currently, a heuristic based on the Gale-Shapley algorithm is being used in the application. The IP methods that we propose are not only interesting theoretically, but may also serve as an alternative solution concept for this practical application, and  other similar applications.% We finish out paper by presenting a simulation for the 2007 data of the Hungarian higher education scheme.
\end{quotation}

\begin{quote}
\end{quote}

%%%%%%%%%%%%%%%%%%%%%%%%%%%%%%%%%%%%%%%%%%%%%%%%%%%%%%%%%%%%%%%%%%%%%%%%%

\section*{Introduction}

Gale and Shapley \cite{GS62amm} introduced and solved the college admissions problem, which generated a broad interdisciplinary research field in mathematics, computer science, game theory and economics\footnote{The 2012 Nobel-Prize in Economic Sciences has been awarded to Alvin Roth and Lloyd Shapley for the theory of stable allocations and the practice of market design.}. The Hungarian higher education admission scheme is also based on the Gale-Shapley algorithm, but it is extended with a number of heuristics since the model contains some special features. In this paper we will study the possibility of modelling these special features with integer programming techniques.

In the Hungarian higher education matching scheme (see a detailed description in \cite{www.matchinginpractice_hun_uni} and \cite{Biro08tr}), the students apply for programmes. However, for simplicity, we will refer to the programmes as colleges in our models. The first special feature of the application is the presence of ties, and the solution concept of stable score-limits. According to the Hungarian admission policy, when two applicants have the same score at a programme then they should either both be accepted or rejected by that programme. The solution of stable score-limits ensures that no quota is violated, since essentially the last group of students with the same score that would cause a quota violation is always rejected. A set of stable score-limits always exists, and a student-optimal solution can be found efficiently by an extension of the Gale-Shapley algorithm, as shown in \cite{BK13cejor}. This method is the basis of the heuristic used in the Hungarian application.

The second and third special features studied in this paper are the lower and common quotas. A university may set not just an upper quota for the number of admissible students for a programme, but also a lower quota. A violation of this lower quota would imply the cancellation of the programme. Furthermore, a common upper quota may be also introduced for a set of programmes, to limit the number of students admitted to a faculty, at a university or nationwide with regard to the state-financed seats in a particular subject. These concepts were studied in \cite{BFIM10tcs}, where the authors showed that each of these special features makes the the college admission problem NP-hard, as present in the Hungarian application. Finally, students can apply for pairs of programmes in case of teacher studies. This possibility was reintroduced to the scheme in 2010. This problem is closely related to the Hospitals / Residents problem with Couples, where couples may apply for pairs of positions. The latter problem is also known to be NP-hard \cite{Ronn90je}, even for so-called consistent preferences \cite{MM12joco}, and even for a specific setting present in Scotland \cite{BIS11jea} where hospitals have common rankings.

The polytope of stable matchings was described in a number of papers for the stable marriage problem \cite{Rothblum92mp}, \cite{RRV93mor}, and for the college admissions problem \cite{BB00mp}, \cite{Fleiner03mss} and \cite{STQ06mor}. For these classical models, since the extremal point of the polytopes are integral, and thus correspond to stable matchings, one could always use a linear programming solver to compute stable solutions (although the computation of solutions can also be done efficiently using the Gale-Shapley algorithm).

However, by introducing even just one special feature, the existence of a stable matching can no longer be guaranteed, and the problem of finding a stable solution can even be NP-hard. In such cases it may be worth investigating
integer programming techniques for solving these problems in theory and in practice as well. To the best of our knowledge there have been only two recent studies of this kind so far. In the first study Kwanashie and Manlove \cite{KM13wp} investigated the problem of finding a maximum size weakly stable matchings for college admissions problems with ties, a problem known to be NP-hard, and motivated by the Scottish resident allocation scheme\footnote{The same problem has also been investigated in a master's thesis \cite{Podhradsky10}.}. In the other paper \cite{BMcBM14wp} the above mentioned matching with couples problem has been studied. As we have already mentioned, one of the four special features of the Hungarian higher education scheme, namely the presence of paired applications, has a close connection to the problem of couples. However, the remaining three special features studied in our paper, the stable score-limits, the lower and common quotas have not been investigated from this perspective.

It is interesting to note that whilst we are not aware of any large scale application for two-sided stable matching markets, except two minor examples\footnote{In a famous study Roth \cite{Roth91aer} analysed the nature and the long term success of a dozen resident allocation schemes established in the UK in the late seventies. He found that two schemes produced stable outcomes and both of them remained in use. From the remaining six ones, that did not always produce stable matchings, four were eventually abandoned. The two programs that were not always produced stable solutions but yet remained used were based on linear programming techniques and has been operated for the two smallest market. \"Unver \cite{Unver01jedc} studied these programs and the possible reasons of their survival in detail.}, integer programming is the standard technique used for kidney exchange programs \cite{ABS07acm-ec}, \cite{RSU07aer} and \cite{MO12sea}.

Finally, we would like to highlight that the models and solution techniques presented in this college admission context may well be useful for other applications too. Two important other applications are immediately apparent. Firstly, controlled school choice \cite{AE07wp}, where the policy makers might want to improve the social-ethnic diversity of the schools by setting different quotas for some types of students. The other example is the resident allocation program, as used in Japan \cite{KK12aer}, where both lower and upper quotas can be requested as regional caps to ensure a better coverage in health care services in all geographic areas with regard to each medical specialty.

In section 1 we describe a basic model for the classical College Admissions problem that will be the basis of our extended models. In section 2 we consider the College Admissions problem with ties and describe two integer linear programs for finding a stable set of score-limits. The first model uses the objective function to achieve stability and thus leads to the student-optimal stable set of score-limits. The second model describes all the stable sets of score-limits using an extended IP model. In section 3 we formulate an IP for describing the College Admissions problem with lower quotas, we provide some useful lemmas that can speed up a solution, and we also extend the model for the case where the lower quotas are established for sets of colleges. In section 4 we study the College Admissions problem with common quotas and we give an integer programming model that describes the set of stable solutions. In section 5 we briefly describe the special feature of paired applications. Finally, in section 6, we analyse the possibility and difficulties of formulating general IP models to describe the stable solutions when some combination of special features are present in the application.

\section{A model for the classical College Admissions problem}

Our basic model is an extension of the Rothblum \cite{Rothblum92mp} model (analysed also in \cite{RRV93mor}). This model has been described in \cite{BB00mp}.

Let $A=\{a_1, \dots a_n\}$ be the set of applicants and $C=\{c_1, \dots , c_m\}$ the set of colleges. Let $u_j$ denote the upper quota of college $c_j$. Regarding the preferences and priorities, let $r_{ij}$ denote the rank of college $c_j$ in $a_i$'s preference list, meaning that $a_i$ prefers $c_j$ to $c_k$ if $r_{ij}<r_{ik}$. Let $s_{ij}$ be an integer representing the score of $a_i$ at college $c_j$, meaning that $a_i$ has priority over $a_k$ at college $c_j$ if $s_{ij}>s_{kj}$, where $\bar{s}$ is the maximum possible score. We denote the set of applications by $E$.

We introduce binary variables $x_{ij}\in \{0,1\}$ for each application coming from $a_i$ to $c_j$, as a characteristic function of the matching, where $x_{ij}=1$ corresponds to the case when $a_i$ is assigned to $c_j$. The feasibility of a matching can be ensured with the following two sets of constraints.

\begin{equation}%\tag{$F_i^a$}
\label{eq:applicant_feasible}
\sum_{j: (a_i,c_j)\in E}x_{ij}\leq 1 \mbox{ for each } a_i\in A
\end{equation}

\begin{equation}%\tag{$F_j^c$}
\label{eq:college_feasible}
\sum_{i: (a_i,c_j)\in E}x_{ij}\leq u_j \mbox{ for each } c_j\in C
\end{equation}

Here, \eqref{eq:applicant_feasible} implies that no applicant can be assigned to more than one college, and \eqref{eq:college_feasible} implies that the upper quotas of the colleges are respected.

One way to enforce the stability of a feasible matching is by the following constraint.

\begin{equation}%\tag{$S_{ij}$}
\label{eq:stable}
 \left(\sum_{k: r_{ik}\leq r_{ij}} x_{ik}\right)\cdot u_j + \sum_{h: (a_h,c_j)\in E, s_{hj}>s_{ij}}x_{hj}\geq u_j \mbox{ for each }(a_i,c_j)\in E
\end{equation}

Note that for each $(a_i,c_j)\in E$, if $a_i$ is matched to $c_j$ or to a more preferred college than the first term provides the satisfaction of the inequality. Otherwise, when the first term is zero, then the second term is greater than or equal to the right hand side if and only if the places at $c_j$ are filled with applicants with higher scores.\\

\textbf{Remark 1:} When we have ties in the priorities (due to equal scores), then the following modified stability constraints, together with the feasibility constraints \eqref{eq:applicant_feasible} and \eqref{eq:college_feasible}, lead to \emph{weakly stable} matchings (in a model also known as Hospitals/Residents problem with Ties).

\begin{equation}%\tag{$S_{ij}^T$}
\label{eq:stable_ties}
\left(\sum_{k: r_{ik}\leq r_{ij}} x_{ik}\right)\cdot u_j + \sum_{h: (a_h,c_j)\in E, s_{hj}\geq s_{ij}}x_{hj}\geq u_j \mbox{ for each }(a_i,c_j)\in E
 \end{equation}

Note that weakly stable matchings can have different sizes and the problem of finding a maximum size weakly stable matching is NP-hard (although some good approximation results exist, e.g. in \cite{IM08joco}). See more about this problem, and its solutions by IP techniques in the recent paper of Kwanashie and Manlove \cite{KM13wp}.\\

\textbf{Remark 2:} In the absence of ties, we can get an applicant-optimal (resp.\ an applicant-pessimal) stable solution by setting the objective function of the IP as the minimum (resp.\ maximum) of the following term:
$$ \sum_{(a_i,c_j)\in E}r_{ij}\cdot x_{ij}$$

We remind the reader that these extreme solutions can be obtained with the two versions of Gale and Shapley's deferred acceptance algorithm in linear time \cite{GS62amm}.\\

\textbf{Remark 3:} Ba\"iou and Balinski \cite{BB00mp} proposed an alternative model to describe the stable admission polytope, since the above simple integer program may admit fractional solutions as extreme point. See also Sethuraman et al.\ \cite{STQ06mor} about the alternative model. Fleiner \cite{Fleiner03mss} provided a different description for the stable admission polytope.

\section{Stable score-limits}

The use of score-limits (or cutoff scores) is very common in college admission systems. The applicants have a score at each place to which they are applying and they are ranked according to these scores by the colleges. The solution is announced in terms of score-limits, each college (or a central coordinator) announces the score of the last admitted student, and each student is then admitted to her most preferred place on her preference list where she achieved the score-limit. See more about the Irish, Hungarian, Spanish and Turkish applications in \cite{BK13cejor}. The score-limits can be seen as a kind of dual solution of a matching, or prices in a competitive equilibrium. Azevedo and Leshno \cite{AL12wp} analysed this phenomenon in detail.

In this section we first develop a basic model for the classical College Admissions problem by using score-limits. Then we discuss the case when ties can appear due to students applying to a place with the same score, as happens in Hungary. We show how this extended setting can be described with a similar IP both with and without the use of an objective function.

\subsection{Stable score-limits with no ties}

If we are given a stable matching for a College Admissions problem then we can define a stable set of score-limits by keeping the following requirements. Each student must meet the score-limit of the college where she is admitted, and no student meets the score-limit of a college that rejected her. These two requirements imply that every student is admitted to the best place where she achieved the score-limit. Finally, to ensure that no student is rejected if a quota was not filled we shall require the score-limit of each unfilled college to be minimal.

To describe this solution concept with an integer program, we introduce new variables for the score-limits. Let $t_j$ be the score-limit at college $c_j$, where $0\leq t_j\leq \bar{s}+1$. The feasibility constraints \eqref{eq:applicant_feasible} and \eqref{eq:college_feasible} remain the same, we only need to link the score-limits to the matching and establish the new stability conditions as follows.

\begin{equation}%\tag{$S_{ij}^{sc}$}
\label{eq:score_stable_college}
t_j\leq \left(1-x_{ij}\right)\cdot (\bar{s}+1) +s_{ij} \mbox{ for each } (a_i,c_j)\in E
\end{equation}
and
\begin{equation}%\tag{$S_{ij}^{sa}$}
\label{eq:score_stable_applicant}
s_{ij}+1\leq t_j +\left(\sum_{k: r_{ik}\leq r_{ij}}x_{ik}\right)\cdot (\bar{s}+1) \mbox{ for each } (a_i,c_j)\in E
\end{equation}

Here, \eqref{eq:score_stable_college} implies that if an applicant is admitted to a college then she achieved the score-limit of that college. The other constraint, \eqref{eq:score_stable_applicant}, ensures that if applicant $a_i$ is not admitted to $c_j$ then either her score at $c_j$ is lower than the score-limit, $t_j$, or she is admitted to a college that she preferred.

Finally, we need to ensure that each college that could not fill its quota has a minimal score-limit. We introduce an indicator variable $f_j\in\{0,1\}$ for each college $c_j$ which is equal to zero if the college is unfilled, by using the following constraint.

\begin{equation}%\tag{$S_{ij}^{sc}$}
\label{eq:score_stable_filled1}
f_j\cdot u_j\leq \sum_{i: (a_i,c_j)\in E}x_{ij} \mbox{ for each } c_j\in C
\end{equation}

Then the following constraint ensures that if a college is unfilled then its score-limit is minimal.

\begin{equation}%\tag{$S_{ij}^{sc}$}
\label{eq:score_stable_filled2}
t_j\leq f_j(\bar{s}+1) \mbox{ for each } c_j\in C
\end{equation}

We summarise the above statements in the following Theorem.

\begin{theorem}
The stable matchings and the related stable sets of score-limits of a College Admissions problem correspond to the solutions of the integer linear program consisting of the feasibility conditions \eqref{eq:applicant_feasible}, \eqref{eq:college_feasible}, the stability conditions \eqref{eq:score_stable_college}, \eqref{eq:score_stable_applicant} and conditions \eqref{eq:score_stable_filled1}, \eqref{eq:score_stable_filled2}.
\end{theorem}

\begin{proof}
A matching is feasible if and only if the corresponding solution satisfy the feasibility constraints. Condition \eqref{eq:score_stable_college} implies that if an applicant is admitted to a college then she achieved the score-limit of that college, and \eqref{eq:score_stable_applicant} implies that she has not achieved the score-limit of any college that she prefers to her assignment. Therefore, these two conditions are satisfied if and only if every applicant is admitted to the best place in her list where she achieved the score-limit. Finally, \eqref{eq:score_stable_filled1} and \eqref{eq:score_stable_filled2} ensure that no college can have positive score-limit, and therefore no college can reject any applicant, if its quota is not filled.
\end{proof}

\subsection{Stable score-limits with ties}

The problem with ties has been defined in \cite{Biro08tr} and studied in \cite{BK13cejor} and \cite{FJ14}. Ties can appear, as the scores of the applicants might be equal at a college, and these ties are never broken in the Hungarian application. Therefore a group of students with the same score are either all accepted or all rejected. In the Hungarian application the upper quotas are always satisfied, and the stability is defined with score-limits (cutoff scores) as follows. For a set of score-limits, each applicant is admitted to the first place in her list where she achieves the score-limit. A set of score-limits is feasible if no quota is violated. It is stable if no score-limit can be lowered at any college without violating its quota, while keeping the other score-limits unchanged. When no ties occur then this definition is equivalent to the original Gale-Shapley one. A stable set of score-limits always exists, and may be found using a generalised Gale-Shapley algorithm. Moreover the applicant proposing version leads to an applicant-optimal solution, where each of the score-limit at each college is as small as possible, and a similar statement applies for the college proposing-version (see \cite{Biro08tr}, \cite{BK13cejor} and \cite{FJ14} for details).

Here we present an IP formulation to find an applicant-optimal set of score-limits with respect to the Hungarian version (where no upper quota may be violated) this is called \emph{$H$-stable set of score-limits} in \cite{BK13cejor}. The feasibility constraints \eqref{eq:applicant_feasible} and \eqref{eq:college_feasible} remain the same as in the previous model, and also the two requirements regarding the score-limits, expressed in constraints \eqref{eq:score_stable_college}, \eqref{eq:score_stable_applicant}. However, we cannot require the unfilled colleges to have minimal score-limits in this model, since an unfilled seat might be created by a tie. We describe two possible solutions for this problem. The first is the use a simple objective function as follows.

\begin{equation}
%\tag{OBJ}
\label{objective}
\min \sum_{j=1\dots m} t_j
\end{equation}

The above objective function is necessary to ensure the stability condition, that is no college can decrease its score limit without violating its quota, supposing that the other score-limits remain the same. To summarise, we state and prove the correctness of the integer program as follows.\\

\begin{theorem}
Feasibility conditions \eqref{eq:applicant_feasible}, \eqref{eq:college_feasible} and stability conditions \eqref{eq:score_stable_college}, \eqref{eq:score_stable_applicant} together with the objective function \eqref{objective} comprise an integer linear program such that the optimal solution of this IP corresponds to the applicant-optimal stable set of score-limits.
\end{theorem}

\begin{proof}
The feasibility constraints ensure that any binary solution of the IP corresponds to a feasible matching, where each applicant is admitted to at most one college and no quota is violated at any college. Conditions \eqref{eq:score_stable_college} and \eqref{eq:score_stable_applicant} ensure that each applicant is admitted to the first place in her preference list where she achieved the score-limit. In particular, \eqref{eq:score_stable_college} implies that if $a_i$ is admitted to $c_j$ then she must have reached the score-limit of $c_j$, and \eqref{eq:score_stable_applicant} implies that if $a_i$ is not admitted to $c_j$ or any better college of her preference (i.e.\ when the second term of the right hand side is zero) then $a_i$ could not achieve the score-limit of $c_j$. Finally, the objective function ensures that no college can decrease its score-limit (without violating its quota). However, this objective function also implies that the solution of the IP must correspond to the applicant-optimal stable set of score-limits, since for any other stable set of score-limits at least one college would have a higher score-limit and neither could have lower, so the sum of the score-limits would not be minimal.
\end{proof}

When we want to describe all sets of stable score-limits (and not only the applicant-optimal one) then we have to replace the objective function with some additional conditions. We can do that by introducing some new variables, as described in the next subsection.

\subsection{Stable score-limits with ties and free objective function}

First, we introduce a binary variable $y_j$ for each college $c_j$ that is equal to one when $t_j$ is positive (i.e.\ when there are some applicants rejected from $c_j$), and otherwise it is zero.

\begin{equation}
\label{zero}
t_j\leq (\bar{s}+1)y_j \mbox{ for each } c_j\in C
\end{equation}

Then, for each application $(a_i,c_j)$, we define a new binary variable, $d_{ij}$, that can be equal to one if $a_i$ desires $c_j$ compared to her actual match and $a_i$ would meet the admission criteria at $c_j$ if the score-limit at $c_j$ was decreased by one, where $m$ denotes the number of colleges.

\begin{equation}
\label{desires}
\sum_{r_{ik}\geq r_{ij}}d_{ik}\leq (1-x_{ij})m \mbox{ for each } (a_i,c_j)\in E
\end{equation}

\begin{equation}
\label{deserves}
t_j-1\leq (1-d_{ij})\bar{s}+s_{ij} \mbox{ for each } (a_i,c_j)\in E
\end{equation}

With the help of the new variables, we can now describe the stability condition of the score-limits as follows.

\begin{equation}
\label{stable with no obj}
(u_j+1)(1-y_j)+\sum_{i:(a_i,c_j)\in E} (x_{ij}+d_{ij})\geq u_j+1 \mbox{ for each } c_j\in C
\end{equation}

\begin{theorem}
Feasibility conditions \eqref{eq:applicant_feasible}, \eqref{eq:college_feasible}, stability conditions \eqref{eq:score_stable_college} and \eqref{eq:score_stable_applicant} and conditions \eqref{zero}, \eqref{desires}, \eqref{deserves} for the new variables $d_{ij}$, together with a new stability condition \eqref{stable with no obj} comprise an integer linear program such that each feasible integer solution corresponds to a stable set of score-limits.
\end{theorem}

\begin{proof}
Again, the feasibility conditions ensure that the corresponding matching is feasible, if and only if the assignment of values to the variables in the IP is feasible. Similarly to the previous model, \eqref{eq:score_stable_college} implies that if an applicant is admitted to a college then she achieved the score-limit of that college, and \eqref{eq:score_stable_applicant} implies that each applicant is admitted to the best available place in her list, if admitted somewhere. Now we shall prove that the remaining four sets of conditions are satisfied if and only if the set of score-limits is stable, i.e.\ when no college can decrease its score-limit without violating its quota.
Suppose first that we have a stable set of score-limits. We assign values to all variables in the IP model appropriately and we prove that the constraints are satisfied. So let $t_j$ be the score-limit at college $c_j$ and we set $y_j$ to be one if $t_j$ is positive. Let $d_{ij}$ be equal to one for each applicant $a_i$ who would prefer to be matched to college $c_j$ than her current partner and who would also meet the score-limit at $c_j$ if it was decreased by one (i.e.\ $s_{ij}=t_j-1$), and we set all the other $d_{ij}$ variables to be zero. When doing so, we satisfy conditions \eqref{desires} and \eqref{deserves}, since \eqref{desires} is satisfied when no $d_{ik}$ is equal to one if $a_i$ prefers her current match $c_j$ to $c_k$, and \eqref{deserves} is satisfied if $a_i$ meets the score-limit at $c_j$ if it is decreased by one.  The stability of the set of score-limits means that no college (with a positive score-limit) can decrease its score-limit without violating its quota. This means that if a college has a positive score-limit, so $y_j=1$, then if it decreased its score-limit by one then the new students admitted were exactly those for whom the corresponding variable $d_{ij}$ is equal to one. The violation of the quota implies that \eqref{stable with no obj} must be satisfied.
To prove the converse, let us suppose that have an assignment of values to the variables in our IP model such that all the constraints in our IP model are satisfied. We shall prove that the set of score-limits as defined by variables $t_j$ is stable, that is, for each college $c_j$ either its score-limit is zero or the decrease of its score-limit would cause a quota violation. When $t_j$ is positive then $y_j$ must be zero, so $\sum_{i:(a_i,c_j)\in E} (x_{ij}+d_{ij})\geq u_j+1$. Since $d_{ij}$ can be one only if $a_i$ both desire and deserves $c_j$ when $t_j$ is decreased to $t_j-1$, this means that the quota at $c_j$ would be indeed violated when the score-limit would be decreased by one. (Note that $y_j$ does not necessarily have to be zero when $t_j$ is zero and  $d_{ij}$ does not necessarily have to be one when $a_i$ both deserves and desires $c_j$.)
\end{proof}

Therefore now we can compute both the student-optimal and student-pessimal stable score-limits, by setting the objective function as described in Remark 2. These extremal solutions can be also computed efficiently by the two generalised versions of the Gale-Shapley algorithm, as shown in \cite{BK13cejor}.

We note that in \cite{BK13cejor} there was another stability definition, the so-called L-stability, that is based on a more relaxed admission policy, namely when the last group of student with the same score with whom the quota would be violated are always accepted. In this paper we are focusing on the setting that is present in the Hungarian application, so we do not deal with L-stability, but it would be possible to describe an IP model for that version as well in a similar fashion.

%\textbf{Question 1:} How can we compute the so-called L-stable student optimal and the L-stable student pessimal solutions, that were defined in \cite{BK13cejor})?

%\textbf{Question 2:} Based on the example of Ba\"iou and Balinski \cite{BB00mp} (Figure 3 on page 430), it is easy to show that our IP does not describe the polytope of stable set of score-limits, as our polytope has non-integer extreme points. The question is whether there is any simple description for the polytope of sets of stable score-limits.

%Write down this example!!!

\section{Lower quotas}

In this section we extend the classical College Admissions problem with the possibility of having lower quotas set for the colleges. After developing an integer program for finding a stable solution for this problem we describe the current heuristic used in Hungary and we also provide some Lemmas that can speed up the solution of the IP. Finally we discuss the possibility of having lower quotas for sets of colleges.

\subsection{College Admissions problem with lower quotas}

This problem has been defined in \cite{BFIM10tcs}. In addition to the College Admissions model, here we have lower quotas as well. Let $l_j$ be the lower quota of college $c_j$. In a feasible solution a college can either be closed (in which case there is no student assigned to there), or open, when the number of students admitted must be between its lower and upper quotas. To describe this feasibility requirement, besides keeping \eqref{eq:applicant_feasible}, we modify \eqref{eq:college_feasible} as follows. We introduce a new binary variable, $o_j\in \{0,1\}$ for each college $c_j$, where $o_j=1$ corresponds to the case when the college is open.

\begin{equation}%\tag{$F_j^{cl}$}
\label{eq:lower_feasible}
o_j\cdot l_j\leq \sum_{i: (a_i,c_j)\in E} x_{ij} \leq o_j\cdot u_j \mbox{ for each } c_j\in C
\end{equation}

The above set of constraints together with \eqref{eq:applicant_feasible} ensure the feasibility of the matching.

The stability of a solution requires the lack of traditional blocking pairs for open colleges, and the lack of blocking groups for closed colleges. The latter means that there cannot be at least as many unsatisfied students (unassigned or assigned to a less preferred place) at a college as the lower quota of that college. The stability conditions can be enforced with the following conditions.

\begin{equation}%\tag{$S_{ij}^{l1}$}
\label{eq:lower_stable1}
 \left(\sum_{k: r_{ik}\leq r_{ij}}x_{ik}\right)\cdot u_j + \sum_{h: (a_h,c_j)\in E, s_{hj}>s_{ij}}x_{hj}\geq o_j\cdot u_j \mbox { for each } (a_i,c_j)\in E
\end{equation}

\begin{equation}%\tag{$S_{j}^{l2}$}
\label{eq:lower_stable2}
\sum_{i: (a_i,c_j)\in E} \left[1-\sum_{k: r_{ik} < r_{ij}}x_{ik}\right] \leq (1-o_j)\cdot (l_j-1) + o_j\cdot n \mbox { for each } c_j\in C
\end{equation}

The first condition implies the usual stability for open colleges, whilst the second condition implies the group-stability for closed colleges. Below we give a formal definition for the College Admission problem with lower quota and a proof of the above description in the following Theorem.

\begin{theorem}
The feasibility conditions \eqref{eq:applicant_feasible} and \eqref{eq:lower_feasible} together with the stability conditions \eqref{eq:lower_stable1} and \eqref{eq:lower_stable2} form an integer program such that its solutions correspond to the stable matchings of a college admissions problem with lower quotas.
\end{theorem}

\begin{proof}
A solution of the IP satisfies the feasibility conditions \eqref{eq:applicant_feasible} and \eqref{eq:lower_feasible} if and only if the corresponding matching is feasible, i.e., no student is admitted to more than one college and the lower and upper quotas are respected in each open college. Regarding stability, condition \eqref{eq:lower_stable1} is redundant if the college is closed and implies the pairwise stability condition for any open college. Condition \eqref{eq:lower_stable2} is redundant for open colleges, and enforces group-stability for any closed college. To show the latter we shall see that the right hand side of the constraint is $l_j-1$ for a closed college $c_j$ and on the left hand side those applicants of $c_j$ are counted who are not admitted to any preferred place according to their preferences, so these are the students who would be happy if $c_j$ would be open and admit them.
\end{proof}

\subsection{Heuristics}

The problem of finding a stable matching for the college admissions problem with lower quotas is proved to be NP-hard \cite{BFIM10tcs}. In the Hungarian application, where lower quotas can be set for any programme, the following heuristic is used with regard to this special feature. First, the applicant-proposing Gale-Shapley algorithm produces a stable matching where some lower quotas might be violated. The heuristic closes one programme, where the ratio of the number of students admitted and lower quota is minimal, and then the applicant-proposing Gale-Shapley algorithm continues by letting the rejected students (whose assigned programme has just been cancelled) apply to their next choices. This heuristic runs in linear time in the number of applications, and it produces the applicant-optimal stable matching for the remaining open colleges. However, as it was illustrated in \cite{BFIM10tcs}, this heuristic can easily produce unstable outcomes even when the problem is solvable. With the IP technique, however, the IP model should guarantee to find a stable solution, whenever it exists. The following lemmas can help in speeding up the solver.

\begin{lemma}\label{compare}
Let $I$ be in instance of College Admission problem and $I'$ be a reduced market where a college is missing. Then the number of students admitted to any college in $I'$ must be at least as many as the number of students admitted in $I$. Moreover, when we compare the student-optimal (resp.\ college optimal) stable matchings in $I$ and $I'$ then each student gets admitted to a college at least as good in $I$ then in $I'$.
\end{lemma}

\begin{proof}
By the Rural Hospitals theorem (\cite{GS85dam}, \cite{Roth84jpe} and \cite{Roth86e}) we know that in the College Admissions problem each college admits the same number of students in every stable matching. Suppose first that we consider the student-optimal stable matching in $I$. When we remove a college then we can invoke the proposal-rejection sequence used in the Gale-Shapley student proposing algorithm and obtain a new stable matching where each college admits at least as many students as before and each student is admitted to the same or a worse place (or nowhere). Suppose now that we consider the college optimal stable matching for $I'$. When we add back the missing college and restart the college-proposing Gale-Shapley process then we will obtain a new stable solution that is at least as good for each student as the previous one.
\end{proof}

\begin{lemma}\label{lower1}
The colleges that reach their lower quotas in the stable solutions of a College Admissions problem with no lower quotas must be open in every stable solution where lower quotas are respected.
\end{lemma}

\begin{proof}
Suppose now for a contradiction that a college $c_j$ that reaches its lower quota for the original problem where all the colleges are open, denoted by $I$, but there is a stable solution where $c_j$ is closed. Let $X$ denote the set of closed colleges in this solution, where $c_j\in X$, and let us denote the submarket where colleges in $X$ are closed by $I_X$. For the market where every college in $X$ but $c_j$ is closed, denoted by $I_{X\setminus j}$, $c_j$ must still reach its lower quota by Lemma \ref{compare}. Furthermore, when we consider the applicant-optimal stable matching for $I_{X\setminus j}$, when we remove $c_j$ and conduct the applicant proposing deferral-acceptance process of Gale-Shapley, as in the Hungarian application, in the resulting stable matching for $I_X$ the students who were previously matched to $c_j$ are all worse off. Therefore for any pairwise stable matching in $I_X$ these students would block the matching with college $c_j$, since thew all prefer $c_j$ to their current match and they are at least as many as the lower quota of $c_j$.
\end{proof}

\begin{lemma}\label{lower2}
Suppose that $X$ is the set of colleges that do not reach their lower quotas in the stable solutions with no lower quotas. Given a college $c_j$ of $X$, if all the colleges in $X$ but $c_j$ are closed and $c_j$ still does not achieve its lower quota then $c_j$ must be closed in any stable solution with lower quotas.
\end{lemma}

\begin{proof}
As we have seen in Lemma \ref{lower1}, no college outside $X$ may be closed in any stable solution, therefore for any stable solution a subset of $X$ should be closed. Suppose for a contradiction that we have a stable matching with lower quotas where $Y\subset X$ is the set of closed colleges and $c_j\notin Y$, meaning that $c_j$ also reaches its lower quota in this matching. Let $Y'=X\setminus c_j$, our assumption is that $c_j$ does not reach its lower quota in the stable matchings when the colleges in $Y'$ are closed. Since $Y\subseteq Y'$, Lemma implies that $c_j$ has at least as many students admitted in the stable matchings for $I_Y$ (the market when the colleges in $Y$ are closed) than for $I_{Y'}$, a contradiction.
\end{proof}

With the help of lemmas \ref{lower1} and \ref{lower2} we can iteratively find some colleges which must be open and perhaps also some that must be closed in any stable solution, thus reduce the number of variants of our model, as follows. First we run the Gale-Shapley algorithm without lower quotas, and then we set each college that reached its lower quota to be open. Let us denote this set of colleges by $X_1$. In the second step we check each college in $C\setminus X_1$ whether it can be open in any stable solution, as described in Lemma \ref{lower2}. That is, for each  $c_j\in C\setminus X_1$ we close all the colleges $C\setminus (X_1\cup c_j)$, run the Gale-Shapley and check whether $c_j$ reaches its lower quota. If not then we set $c_j$ to be closed. Let $Y_1$ denote the set of colleges that were found to be necessary to close in the second step. If $Y_1$ is nonempty then we repeat the first step, we run again the Gale-Shapley algorithm with no lower quotas and without colleges $Y_1$. If, in addition to $X_1$, some new colleges also achieve their lower quotas in the reduced market then we add them to $X_1$ and get a larger set of colleges, $X_2$, that must be open in every stable solutions. In case $X_2$ is larger than $X_1$ then we shall repeat the second step of our process with respect to $X_2$. Again, if we find any new college that must be closed then we increase the set $Y_1$ to $Y_2$. We repeat this process until $X_t=X_{t+1}$ or $Y_t=Y_{t+1}$ for any $t$, and then we stop.

\subsection{Lower quotas for sets of colleges}

We shall note that in the Hungarian application the lower quotas are set for pairs of programmes (actually for the same programme, just with separate quotas for state-financed and privately financed students). If the lower quota for a set of colleges is not met then all the colleges should be closed. This motivates the extended model where the lower quotas can be applied for sets of colleges. In this case we need only to introduce a new binary variable $o_p$ for each set of colleges $C_p$ with common lower quota $l_p$ that we associate with the other indicator variables of individual colleges as follows.

\begin{equation}%\tag{$S_{ij}^{l1}$}
\label{eq:common_lower_feasible}
 o_p\cdot n_p\leq \sum_{j: c_j\in C_p} o_{j} \leq o_p\cdot n_p \mbox{ for each } C_p
\end{equation}

where $n_p$ is the number of colleges in $C_p$. We shall then set the feasibility conditions \eqref{eq:lower_feasible} for the sets of colleges with common lower quotas as follows.

\begin{equation}%\tag{$F_j^{cl}$}
\label{eq:lower_feasible_set}
o_p\cdot l_p\leq \sum_{i: (a_i,c_j)\in E} x_{ij}  \mbox{ for each } c_j\in C
\end{equation}

To define stability in this setting is not easy though. In section 6 we discuss some complications and potential solutions, basically by suggesting to drop the group-stability criteria in the extended models, like this.

\section{Common quotas}

This problem has also been defined in \cite{BFIM10tcs}. For each set of colleges $C_p\subseteq C$ the coordinator of the scheme may set a common upper quota, $u_p$, meaning that the total number of students admitted to colleges in $C_p$ cannot exceed this quota. Therefore, the set of feasibility constraints, \eqref{eq:applicant_feasible} and \eqref{eq:college_feasible}, has to be extended with some new constraints enforcing the common quotas, as follows.

\begin{equation}%\tag{$F_j^{cc}$}
\label{eq:common_college_feasible}
\sum_{(a_i,c_j)\in E, c_j\in C_p}x_{ij}\leq u_p \mbox{ for each } C_p\subseteq C
\end{equation}

Regarding stability, first of all we have to suppose that any two colleges, $c_j$ and $c_k$, that belong to a set of colleges $C_p$ with a common quota must rank their applicants in the same way. In particular, in the Hungarian application any student $a_i$ has the same score at such colleges (i.e.\ programmes in Hungary) with common quota, so $s_{ij}=s_{ik}$ holds. (In a more general model, we should suppose to have a specific scoring for each set of colleges with common quota, which is in agreement with the individual scorings of the colleges belonging to this set. For instance, we could have a score $s_{ij}^p$ for each application associated to a set of colleges $C_p$ with a common quota such that $s_{ij}>s_{lj}$ implies $s_{ij}^p>s_{lj}^p$).

In this setting stability means that if a student $a_i$ is not admitted to a college $c_j$ or to any better college of her preference then either $c_j$ must have filled its quota with better students or there is a set of colleges $C_p$, such that $c_j\in C_p$ and all the $u_p$ places in $C_p$ have been filled with better students than $a_i$. Bir\'o et al.\ \cite{BFIM10tcs} showed that if the sets of colleges with common upper quotas is nested, i.e., when $C_p\cap C_q\neq \emptyset$ implies either $C_p\subset C_q$ or $C_p\supset C_q$, then a stable matching always exists. Moreover, a stable matching can be found efficiently by the generalised Gale-Shapley algorithm and there are applicant and college optimal solutions. However, if the set system is not nested then stable solution may not exist and the problem of finding a stable matching is NP-hard. Interestingly, the Hungarian application involved nested set systems until 2007 when a legislative change modified the structure of the underlying model and made the set system non-nested, with the possibility of having no stable solution and also making the problem computationally hard.

Here, we show that we can express this stability condition with the use of score-limits, in a similar fashion to the method we described in section 2. However, here we need to assume that there are no ties. We set a score-limit $t_p$ for each set of colleges $C_p$ with common quota, which is less than or equal to the score of the weakest admitted student if the common quota is filled, and 0 if the common quota is unfilled in the matching. When describing the model in this way, stability implies that if a student $a_i$ is admitted to college $c_j$ then $s_{ij}\geq t_j$ and also $s_{ij}\geq t_p$ for any set of colleges $C_p$ with common quota where $C_p$ includes $c_j$. Furthermore, if $a_i$ is not admitted to $c_j$ or to any better college of her preference then it must be the case that either $s_{ij}<t_j$ or $s_{ij}<t_p$ for some set of colleges $C_p$ with common quota where $C_p$ contains $c_j$. These conditions can be formalised with the following set of conditions, where $q_{j}$ denotes the number of sets of colleges with common quota involving college $c_j$, $\{c_j\}$ also being one of them.

\begin{equation}%\tag{$S_{ijp}^{csc}$}
\label{eq:common_score_stable_college}
t_p\leq \left(1-x_{ij}\right)\cdot (\bar{s}+1) +s_{ij} \mbox{ for each } (a_i,c_j)\in E \mbox{ and } c_j\in C_p
\end{equation}
and
\begin{equation}%\tag{$S_{ijp}^{csa}$}
\label{eq:common_score_stable_applicant}
s_{ij}+1\leq t_p +\left(\sum_{k: r_{ik}\leq r_{ij}}x_{ik}+y_{i}^p\right)\cdot (\bar{s}+1) \mbox{ for each } (a_i,c_j)\in E \mbox{ and } c_j\in C_p
\end{equation}
with
\begin{equation}%\tag{$S_{ij}^{cse}$}
\label{eq:common_score_stable_exception}
\sum_{p:c_j\in C_p} y_{i}^p\leq q_j-1 \mbox{ for each } (a_i,c_j)\in E
\end{equation}

where $y_{i}^p\in \{0,1\}$ is a binary variable. These conditions are needed to establish the links between a matching and the corresponding score-limits. However, for stability we also have to ensure that the score-limits are minimal. In case of strict preferences (i.e., when no two students have the same score at colleges belonging to a set of colleges with a common quota), we can ensure the minimality of the score-limits with the following conditions.

Again, we introduce an indicator variable $f_p$ for each set of colleges $C_p$ which is equal to zero if the common quota of these colleges is unfilled, by using the following constraints.

\begin{equation}%\tag{$S_{ij}^{sc}$}
\label{eq:common_score_stable_filled1}
f_p\cdot u_p\leq \sum_{i: (a_i,c_j)\in E, c_j\in C_p}x_{ij} \mbox{ for each } C_p\subseteq C
\end{equation}

Then we ensure that if a college or a set of colleges with a common quota is unfilled then its score-limit is zero.

\begin{equation}%\tag{$S_{ij}^{sc}$}
\label{eq:common_score_stable_filled2}
t_p\leq f_p(\bar{s}+1) \mbox{ for each } C_p\subseteq C
\end{equation}

We describe and prove the correctness of the IP model in the following Theorem.

\begin{theorem}
Feasibility conditions \eqref{eq:applicant_feasible}, \eqref{eq:college_feasible} and \eqref{eq:common_college_feasible}, with the stability conditions \eqref{eq:common_score_stable_college}, \eqref{eq:common_score_stable_applicant} and \eqref{eq:common_score_stable_exception}, together with \eqref{eq:common_score_stable_filled1} and \eqref{eq:common_score_stable_filled2} describe an integer program such that its solutions correspond to stable matching for the College Admissions problem with common quotas.
\end{theorem}

\begin{proof}
To see the correctness, we have to note first that the matching is feasible if and only if the feasibility constraints \eqref{eq:applicant_feasible}, \eqref{eq:college_feasible} and \eqref{eq:common_college_feasible} are satisfied, and conditions \eqref{eq:common_score_stable_filled1} and \eqref{eq:common_score_stable_filled2} are satisfied if and only if the score-limit of any unfilled college or set of colleges is zero.

Now, suppose first that we have a stable solution and we show that all the stability conditions can be satisfied by setting the variables appropriately. When a common quota $C_p$ is filled we set $t_p$ to be equal to the last admitted applicant at any college included in $C_p$. This ensures that the first set of conditions \eqref{eq:common_score_stable_college} is satisfied.  Let $y_i^p=0$ if $a_i$ does not meet $t_p$ at any college $c_j\in C_p$ where she applied to, and $y_i^p=1$ otherwise, with the exception of set $\{c_j\}$, where we set $y_i^j=0$ if $i$ is admitted to $c_j$ or to a better place. The stability of the matching then implies \eqref{eq:common_score_stable_exception}. Finally, let us consider an application $(a_i,c_j)$ where $c_j\in C_p$. If $a_i$ is admitted to $c_j$ or to a better college then the corresponding constraint \eqref{eq:common_score_stable_applicant} is satisfied, irrespective of the value $y_i^p$. Otherwise, suppose that $a_i$ is not admitted to $c_j$ or to any better place. If $a_i$ does not meet the score-limit $t_p$ then \eqref{eq:common_score_stable_applicant} is satisfied, obviously, and it is also satisfied when she meets $t_p$, since $y_i^p=1$ in that case.

Conversely, suppose that we have a solution for the IP model, and we will show that this ensures the stability of the corresponding matching. If $(a_i,c_j)$ is in the matching then constraints \eqref{eq:common_score_stable_college} imply that $a_i$ achieves the score-limit of $c_j$ and also the score-limit of every set of colleges with common quota containing $c_j$. Finally, suppose that $a_i$ is not admitted to $c_j$ or to any better college of her preference. Since one of the additional variables of form $y_{i}^l$ must be zero, say $y_{i}^p$, the corresponding constraint \eqref{eq:common_score_stable_applicant} implies that the set of colleges $C_p$ containing $c_j$ has a score-limit $t_p$ greater than $s_{ij}$. So the matching is indeed stable.
\end{proof}

Finally we note that if we have ties then we shall use an objective function that minimises the sum of the score-limits or an extended model, similar to the ones described in section 2.

\section{Paired applications}

In the Hungarian application students can apply for pairs of programmes in case of teachers' studies, e.g.\ when they want to become a teacher in both maths and physics. In this setting of a College Admission problem with paired applications stability means that if a student is not admitted to a pair of colleges, or to any better college (or pair of colleges) in her list then either of these colleges must have filled its quota with better applicants. In terms of score-limits, either of these colleges must have a score-limit that this student has not achieved.

This problem is similar to the well-known Hospitals/Residents problem with Couples, where residents apply to pairs of positions (see a survey \cite{BK13igtr}). However, there are some slight differences. Here a student may have both simple and paired applications in her list, thus she can behave both as a single applicant and as a couple at the same time. A paired application in the college admission problem involves two distinct programmes, whilst a couple may apply for a pair of positions at the same hospital in a resident allocation program. So neither of these problems is more general than the other. Nevertheless both problems are NP-hard, since the NP-hardness proof of Ronn \cite{Ronn90je} for the Hospitals/Residents problem with couples was concerned with a special case where each hospital has one place only, and this is also a special case of the College Admissions problem with paired applications.

Note that the case of paired applications can be modeled in a similar way as in the case of common quotas. Each pair of colleges can be seen as an artificial college, and the upper quotas of the original colleges will become common quotas.

For completeness, we describe the IP model for college admissions with paired application.
We introduce new, artificial colleges $C^P$, that are pairs of compatible colleges. Let $c_{jk}\in C^P$ be one compatible pair of colleges, the indicator variable of an application from student $a_i$ to this pair of colleges is denoted by $x_{i(jk)}$, and the rank of this application in $a_i$'s list is denoted by $r_{i(jk)}$. For a more convenient notation, let $e\in E^S$ denote a simple application, and $e\in E^P$ denote a paired application, where the set of all applications is $E=E^S\cup E^P$. With a slight abuse of notation, let $a_i\in e$ mean that application $e$ is coming from applicant $a_i$, and $c_j\in e$ or $c_{(jk)}\in e$ mean that the application is going to college $c_j$ or pair of colleges $c_{(jk)}$, respectively.

The feasibility constraints can now be modified in the following way.

\begin{equation}%\tag{$F_i^{pa}$}
\label{eq:paired_applicant_feasible}
\sum_{e: a_i\in e\in E}x_{e}\leq 1 \mbox{ for each } a_i\in A
\end{equation}

\begin{equation}%\tag{$F_j^{pc}$}
\label{eq:paired_college_feasible}
\sum_{e: c_j\in e\in E^S}x_{e} + \sum_{k=1\dots m}\sum_{e: c_{(jk)}\in e\in E^P}x_{e} \leq u_j \mbox{ for each } c_j\in C
\end{equation}

The stability conditions are expressed with score-limits. There is no change for simple applications.

\begin{equation}%\tag{$S_{ij}^{sc}$}
\label{eq:simple_score_stable_college}
t_j\leq \left(1-x_{ij}\right)\cdot (\bar{s}+1) +s_{ij} \mbox{ for each } (a_i,c_j)\in E^S
\end{equation}
and
\begin{equation}%\tag{$S_{ij}^{sa}$}
\label{eq:simple_score_stable_applicant}
s_{ij}+1\leq t_j +\left(\sum_{e: a_i\in e, r_{e}\leq r_{ij}}x_{e}\right)\cdot (\bar{s}+1) \mbox{ for each } (a_i,c_j)\in E^S
\end{equation}

For paired applications, the following conditions should hold.

\begin{equation}%\tag{$S_{i(jk)}^{psc1}$}
\label{eq:paired_score_stable_college1}
t_j\leq \left(1-x_{i(jk)}\right)\cdot (\bar{s}+1) +s_{ij} \mbox{ for each } (a_i,c_{(jk)})\in E^P
\end{equation}
\begin{equation}%\tag{$S_{i(jk)}^{psc2}$}
\label{eq:paired_score_stable_college2}
t_k\leq \left(1-x_{i(jk)}\right)\cdot (\bar{s}+1) +s_{ik} \mbox{ for each } (a_i,c_{(jk)})\in E^P
\end{equation}
and
\begin{equation}%\tag{$S_{i(jk)}^{psa1}$}
\label{eq:paired_score_stable_applicant1}
s_{ij}+1\leq t_j +\left(\sum_{e: a_i\in e, r_{e}\leq r_{i(jk)}}x_{e}+y_i^{(jk)}\right)\cdot (\bar{s}+1) \mbox{ for each } (a_i,c_{(jk)})\in E^P
\end{equation}
\begin{equation}%\tag{$S_{i(jk)}^{psa2}$}
\label{eq:paired_score_stable_applicant2}
s_{ij}+1\leq t_k +\left(\sum_{e: a_i\in e, r_{e}\leq r_{i(jk)}}x_{e}+(1-y_i^{(jk)})\right)\cdot (\bar{s}+1) \mbox{ for each } (a_i,c_{(jk)})\in E^P
\end{equation}

Here, the first two conditions, \eqref{eq:paired_score_stable_college1} and \eqref{eq:paired_score_stable_college2}, ensure that the score limit is met for each of the colleges where an applicant is admitted with a paired application. The second two conditions, \eqref{eq:paired_score_stable_applicant1} and \eqref{eq:paired_score_stable_applicant2}, imply that if a student $a_i$ is not admitted to a pair of colleges $c_{(jk)}$ or to any better place(s) then it must be the case that she did not achieve the score-limit at one of these colleges. Here again, $y_i^{(jk)}\in \{0,1\}$ is a binary variable.

Finally, similarly to conditions \eqref{eq:common_score_stable_filled1} and \eqref{eq:common_score_stable_filled2}, we have to ensure that only the filled colleges can have positive score-limits.

\begin{equation}%\tag{$S_{ij}^{sc}$}
\label{eq:paired_score_stable_filled1}
f_j\cdot u_j\leq \sum_{e: c_j\in e\in E^S}x_{e} + \sum_{k=1\dots m}\sum_{e: c_{(jk)}\in e\in E^P}x_{e} \mbox{ for each } c_j\in C
\end{equation}

Then we ensure that if a college or a set of colleges with a common quota is unfilled then its score-limit is zero.

\begin{equation}%\tag{$S_{ij}^{sc}$}
\label{eq:paired_score_stable_filled2}
t_j\leq f_j(\bar{s}+1) \mbox{ for each } c_j\in C
\end{equation}

\begin{theorem}
The College Admissions problem with paired applications can be described with the solutions of the IP consisting of conditions \eqref{eq:paired_applicant_feasible} to \eqref{eq:paired_score_stable_filled2}.
\end{theorem}

\begin{proof}
Conditions \eqref{eq:paired_applicant_feasible} and \eqref{eq:paired_college_feasible} are satisfied if and only if the corresponding matching is feasible. Regarding the stability constraints, if we consider a single application then \eqref{eq:simple_score_stable_college} ensures that the applicant has reached the score-limit of the college where she has been admitted, and \eqref{eq:simple_score_stable_applicant} implies that a rejection must have taken place because the applicant has not reached the score-limit of that college. Similarly, for a paired application \eqref{eq:paired_score_stable_college1} and \eqref{eq:paired_score_stable_college2} are satisfied if and only if the applicant has achieved the score-limits of both colleges where she has been admitted in a paired application. Conditions \eqref{eq:paired_score_stable_applicant1} and \eqref{eq:paired_score_stable_applicant2} ensure that if a paired application is rejected, it must be the case that the applicant has not reached the score-limit of either of the colleges in her paired application. To summarise, each applicant is admitted to the best college or pair of colleges on her list where she has achieved the score-limit(s). Finally, conditions \eqref{eq:paired_score_stable_filled1} and \eqref{eq:paired_score_stable_filled2} imply that only those colleges which are filled may reject applications.
\end{proof}

\section{Combining the models into a single IP}

In the four previous sections we have developed IP models to deal with each of the four special features which are present in the Hungarian application. However, all of these special features are present in the application simultaneously, and so to provide a solution for real data we need to create a combined model which incorporates all of the constraints. However, this task is not easy, since not only the constraints have to be adjusted, but sometimes the stability definitions may also contradict each other. Since paired applications can be seen as a special case of common quotas, as we described in the previous section, we focus on the different combinations of ties, lower and common quotas.

\subsection{Stable score-limits with ties and lower quotas}

The feasibility of the matching is characterised by constraints \eqref{eq:applicant_feasible} and \eqref{eq:lower_feasible}. For stability, we shall use \eqref{eq:score_stable_college} to ensure that a student is only admitted if she achieved the score-limit. We need \eqref{eq:lower_stable2} again for group-stability. Here we need a combination of \eqref{eq:score_stable_applicant} and \eqref{eq:lower_stable1} to enforce stability for open colleges, as follows.

\begin{equation}%\tag{$S_{ij}^{sa}$}
\label{eq:lower_score_stable_applicant}
s_{ij}+1\leq t_j +\left(\sum_{k: r_{ik}\leq r_{ij}}x_{ik} +1-o_j\right)\cdot (\bar{s}+1) \mbox{ for each } (a_i,c_j)\in E
\end{equation}

Finally, either we need to minimise the sum of score-limits with the objective function or to use conditions \eqref{zero}, \eqref{desires}, \eqref{deserves} and \eqref{stable with no obj} to ensure that the decrease of any positive score-limit in the solution would cause the violation of a quota.

\subsection{Stable score-limits with ties and common quotas}

As we have mentioned in section 4, our IP model for common quotas are based on score-limits, so it is not too difficult to reconcile these two features. The issue that we must resolve is that we cannot require the unfilled colleges or sets of colleges to have zero score-limit here, as this might have resulted from the rejection of a group of applicants with the same score. So instead of using \eqref{eq:common_score_stable_filled1} and \eqref{eq:common_score_stable_filled2}, either we have to minimise the sum of the score-limits with the objective function to enforce stability, or we shall develop a model similar to the one described in subsection 2.3 for sets of colleges with common quotas.

\subsection{The difficulty in reconciling lower and common quotas}

The feasibility of the matching is now characterised by constraints \eqref{eq:applicant_feasible}, \eqref{eq:lower_feasible} and \eqref{eq:common_college_feasible}. For stability, we need \eqref{eq:common_score_stable_college} to ensure that a student is only admitted if she achieved the score-limit of the college assigned and also the score-limit for each set of colleges with a common quota that contains this college. For open colleges, the pairwise stability condition can be enforced with the following modification of \eqref{eq:common_score_stable_applicant}:

\begin{equation}
\label{eq:common_lower_score_stable_applicant}
s_{ij}+1\leq t_p +\left(\sum_{k: r_{ik}\leq r_{ij}}x_{ik}+y_{i}^p+1-o_j\right)\cdot (\bar{s}+1) \mbox{ for each } (a_i,c_j)\in E \mbox{ and } c_j\in C_p
\end{equation}
together with \eqref{eq:common_score_stable_exception}.

However, defining group-stability is problematic here. Since even if there is a closed college with more unsatisfied students than its lower quota the admission of these students could lead to the violation of a common upper quota. In which case this block should not be allowed. So combining the constraints for lower and common quotas is challenging. Therefore one might consider the possibility of abandoning the group-stability constraints related to lower quotas. However, in this case, a possible solution would be to close all the colleges. This unsatisfactory scenario could be avoided by setting the objective function such that the number of students admitted becomes the first priority and the minimisation of the total score limits is the second priority.

\subsection{Stable score-limits with ties, lower and common quotas}

As we have described, the concept of stable stable score-limits for College Admissions problem with ties can be reconciled with both lower and common quotas, however the latter two are difficult to combine. Therefore in the practical application the organisers should decide what stability conditions they want to satisfy in the first place. One possibility is to drop the group-stability conditions with regard to lower quotas and to use an appropriate objective function to ensure that the solution is not biased. In order to set the right conditions one should do simulations with the real data to see how the model, and the objective function in particular, may influence the resulting solution.

\section*{Further notes}

As we described in the previous section, the combination of the four special features result in interesting challenges in two ways. First, we need to define appropriate stability criteria when both lower and upper quotas are present. Second, we need to combine the separate integer programmes into a single programme that would result in a suitable solution for the real application.

It would be also important to know whether these IP formulations may be solved within a realistic timescale for such a large scale application as the Hungarian higher education matching schemes, with around 100000 applicants. Our plan is to conduct experiments on real data we have access to, which comes from the 2008 match run of the Hungarian higher education scheme. %However, to test whether the simple IP-s described for each special feature are possible to solve, we will need to filter and manipulate the data, which is already a difficult task.

Furthermore, one could always try to solve our special college admissions problems with other approaches, e.g. with different integer programming formulations. Finally, our models may be useful in other applications as well, such as controlled school choice (see e.g. \cite{AE07wp}), resident allocation with distributional constraints (see e.g. \cite{KK12aer}), or for finding stable solutions with additional restrictions, such as matchings with no Pareto-improving swaps \cite{Irving08joco}.

\end{document}